\documentclass[
  prl,
  twocolumn,
  superscriptaddress,
  citeautoscript,
  showpacs,
  amsart,
  longbibliography
]{revtex4-1}

\usepackage{amsmath,bm,amsfonts}
\usepackage{physics,mathtools}
\usepackage{graphicx}
\usepackage[normalem]{ulem}
\usepackage{titlesec}
\usepackage[colorlinks=true,linkcolor=blue,citecolor=blue]{hyperref}

\newtheorem{thm}{Theorem}%[section]

\newtheorem{cor}[thm]{Corollary}

\newtheorem{rem}{Remark}%[section]

\usepackage[dvipsnames]{xcolor}

\newcommand{\cB}{\mathcal{B}}
\newcommand{\cC}{\mathcal{C}}
\newcommand{\cD}{\mathcal{D}}

\newcommand{\cH}{\mathcal{H}}
\newcommand{\cI}{\mathcal{I}}

\newcommand{\cK}{\mathcal{K}}
\newcommand{\cL}{\mathcal{L}}

\newcommand{\cX}{\mathcal{X}}

\newcommand{\cZ}{\mathcal{Z}}

\def\one{\boldsymbol{1}}

\begin{document}

\title{Stationary states of boundary driven quantum systems: some exact results}

\author{Eric A. Carlen }
\affiliation{Department of Mathematics, Rutgers University, Piscataway, NJ 08854-8019 USA}
\author{David A. Huse}
\affiliation{Department of Physics, Princeton University, Princeton, NJ 08544, USA}
\author{Joel L. Lebowitz}
\affiliation{Departments of Mathematics and Physics, Rutgers University, Piscataway, NJ 08854-8019 USA}

\date{\today}

\begin{abstract}
We study finite-dimensional open quantum systems whose density matrix $\rho$ evolves via a Lindbladian, $\dot{\rho}=-i[H,\rho]+\cD\rho$.  Here $H$ is the Hamiltonian of the 
system and $\cD$ is the dissipator.  We consider the case where the system consists of two parts, the ``boundary'' $A$ and the ``bulk'' $B$, and $\cD$ acts only on $A$, so $\cD=\cD_A\otimes\cI_B$, where $\cI_B$ is the identity superoperator on part $B$.  Let $\cD_A$ be ergodic, so $\cD_A\hat{\rho}_A=0$ only for one unique density matrix $\hat{\rho}_A$.  We show that any stationary density matrix $\bar{\rho}$ on the full system which commutes with $H$ must be of the product form $\bar{\rho}=\hat{\rho}_A\otimes\rho_B$ for some $\rho_B$.  This rules out finding any $\cD_A$ that has the Gibbs measure $\rho_\beta= e^{-\beta H}/Z(\beta)$ as a stationary state with $\beta\neq 0$, unless there is no interaction between parts $A$ and $B$.  We give criteria for the uniqueness of the stationary state $\bar{\rho}$ for systems with interactions between $A$ and $B$.  Related results for non-ergodic cases are also discussed.  
\end{abstract}

\maketitle

%%%%%%%%%%%%%%%%%%%%%%%%%%%%%%%%%%%%%%%%%%%%%%%%%%%%%%%%%%%%%%%%%

\section{Introduction}

There is much current interest, theoretical and experimental, in open quantum systems \cite{LPS}.  These are often quantum systems in contact with macroscopic equilibrium systems which act as thermal reservoirs.  Analyzing in detail the total system, including the reservoirs, is usually too difficult.  Since the quantities of interest are the time-evolution and stationary states %primarily in the behavior of the density matrix 
of the system, its interaction with the reservoirs is commonly idealized by saying that they cause the system to  %having it %.  This can be described as 
evolve %ing %, in an idealized form, 
under the influence of a stochastic quantum process. This leads to a Markovian master equation for %the time evolution of 
the system's density matrix $\rho(t)$.  The requirement of complete 
positivity of the evolution then restricts the form of this equation to a Lindbladian form: 
\begin{equation}\label{Lindblad1}
\frac{\partial\rho}{\partial t} = -i[H,\rho] +\cD\rho ~,  
\end{equation}
where $H$ is the Hamiltonian of the isolated system, and $\cD\rho$ is the “dissipation” caused by its interactions with the reservoirs.  By a theorem of Lindblad \cite{Lin76} and Gorini, Kossakowski and Sudarshan \cite{GKS76}, the generator $\cD$ has the form
\begin{equation}\label{Lindblad2}
\cD\rho= -i[K,\rho] +
\sum_{\alpha=1}^n(L_\alpha\rho L_\alpha^\dagger-\frac{1}{2}\{L_\alpha^\dagger L_\alpha,\rho\})~,
\end{equation}
where $K$ is a self adjoint operator, sometimes called the Lamb shift Hamiltonian. $K$ could be combined with $H$, but for our purposes, it is convenient to keep the effect of the reservoirs on the dynamics clearly separated from the dynamics of the isolated
system.

%It is important to note that there is no canonical separation of $\cD$ into parts that are ``unitary'' and ``dissipative''.   Indeed, for each $\alpha$, define $\widetilde{L}_\alpha := L_\alpha -z_\alpha\cI_\cH$ where $\alpha$ is a complex number and $\cI_\cH$ is the identity on $\cH$.  Then $L_\alpha = \widetilde{L}_\alpha$, and $\cD$ can be expressed in terms of the operators $K$, $\widetilde{L}$, and $\widetilde{L}_\alpha^\dagger$. Defining $G_\alpha$ and $H_\alpha$ to be the self adjoint operators such that $\bar{z}\widetilde{L}_\alpha = G_\alpha+ iH_\alpha$, simple computations yield
%\begin{equation}\label{Lindblad2V}
%\cD\rho= -i[\widetilde{K},\rho] + \sum_{\alpha=1}^n(\widetilde{L}_\alpha\rho \widetilde{L}_\alpha^\dagger -\frac{1}{2}\{\widetilde{L}^\dagger \widetilde{L},\rho\})~.
%\end{equation}
%where $\widetilde{K} = K +\sum_{\alpha =1}^n H_\alpha$. The generator $\cD$ is the same, but now the ``Hamiltonain part'' is different. 

We shall not discuss here the derivation of $K$ and the ``jump operators'' $\{L_\alpha\}$ from the interactions with the reservoirs via a weak coupling limit, as is done by Davies \cite{Da74}, or in the singular coupling limit as in Gorini and Kossakowski \cite{GorKos76}. (See \cite{Pal77} for a treatment of both limits in a common framework).
We shall focus instead on the relations between the properties of $\cD$ and $H$ and the stationary state(s) $\bar{\rho}$ of \eqref{Lindblad1}.  It follows from general results  that there is always at least one stationary $\bar{\rho}$. (See, e.g. \cite{D70} and Section~\ref{S4} of this paper.) The question of uniqueness of such steady states is closely connected with the existence of a strictly positive steady state
\cite{Frig78,yosh} and is discussed  below in Section~\ref{UNIQ}.

As discussed in \cite{LPS,TDKP} and references therein, there are two desiderata for $\cD$: (i) we would like $\cD$ to act only on the ``boundary" degrees of freedom of the system, as do the reservoirs in certain situations of interest, and (ii) we would like the stationary state $\bar{\rho}$ to be unique and to be that of thermal equilibrium at a finite temperature $1/\beta$ %and chemical potential 
set by the reservoir when the system is interacting with only one such reservoir, i.e. $\bar{\rho}=\frac{1}{Z(\beta)}\exp{(-\beta H)}$, or $\bar{\rho}=\frac{1}{Z(\beta, \mu)}\exp{(-\beta (H-\mu N))}$ if the reservoir also sets a chemical potential for total particle (or excitation) number $N$.  These two desired properties are readily realizable for classical systems \cite{rll} but seem incompatible for quantum systems in the particular cases investigated by \cite{LPS,TDKP}.  Here we prove a ``no go'' theorem showing that this incompatibility is indeed the case quite generally. However, infinite temperature ($\beta\rightarrow 0$) steady states of this form with $\beta\mu$ finite do occur (and are unique) for some such models when $[H,N]=0$, as we show.

\section{Local Lindbladians}

We shall consider a general setup where the system of interest can be divided into two parts $A$ and $B$.  The Lindbladian dissipator will couple only to part $A$, which we can consider to be the boundary of the system, while part $B$ is the bulk of the system.  For example, if our system is a finite spin chain, part $A$ could be the first $m$ spins at one or both ends of the chain (e.g. $m$ could be one or two spins), while part $B$ is all the remaining spins.  This type of ``local coupling'' has been discussed in some detail for various systems, such as the well known XXZ or XYZ spin chain models with %particular forms of 
$\cD$ acting on the spins at the ends of the chain \cite{LPS,yosh,TDKP,pros}.  A particular such spin chain example is discussed below.

The Hilbert space $\cH$ of the full system is the direct product of the Hilbert spaces $\cH_A$ and $\cH_B$, each assumed to be of finite dimension more than one.  The Hamiltonian of the %isolated 
system %, which determines its dynamics, 
is assumed to be finite and can be written as
\begin{equation}
H=H_A\otimes\one_B+H_{AB}+\one_A\otimes H_B~,
\end{equation}
where $H_{AB}$ denotes the interaction between subsystems $A$ and $B$ (without loss of generality, we require ${\rm Tr}_A\{H_{AB}\}={\rm Tr}_B\{H_{AB}\}=0$, defining $H_A$ and $H_B$ accordingly).  $H_{AB}=0$ then describes the dynamics of separately isolated $A$ and $B$ systems.  $\one_A$ ($\one_B$) is the identity operator on $A$ ($B$).  %If there is a nonzero Lamb shift Hamiltonian $K$, it acts only on $A$ because the reservoir couples only to $A$, thus $K$ contributes only to $H_A$.

In mathematical terms, we consider the case where the dissipator is of the form
\begin{equation}\label{dissform}
\cD=\cD_A\otimes\cI_B~,
\end{equation}
where the superoperator $\cD_A$ acts only on operators on subsystem $A$, and $\cI_B$ denotes the identity superoperator on subsystem $B$.  This means that the jump operators in \eqref{Lindblad2} all can be written as 
\begin{equation}\label{Lhat}
L_\alpha=\hat{L}_\alpha\otimes\one_B~,  
\end{equation}
in terms of $\{\hat{L}_\alpha\}$ that act only on $\cH_A$, and similarly the Lamb shift Hamiltonian can be written as 
\begin{equation}\label{Khat}
K=\hat{K}\otimes\one_B~.   
\end{equation}  
We assume further (for the moment) that $\cD_A$ is ergodic %and non-transient 
on subsystem $A$, i.e. that there exists a unique  density matrix $\hat\rho_A$ on $\cH_A$ such that
\begin{equation}\label{darhoa}
\cD_A\hat\rho_A=0~.
\end{equation}
%and this $\hat\rho_A$ is positive definite. 
(We also discuss the case where $\cD_A$ is non-ergodic in Appendix~\ref{NONERG} 
below.)

An example of such a $\cD$ is where one replaces at rate $1/\tau$ the density matrix $\rho(t)$ by $\hat\rho_A\otimes\rho_B(t)$, where $\hat{\rho}_A$ is a specified density matrix of subsystem $A$ and $\rho_B(t)=\tr_A[\rho(t)]$. That is,
\begin{equation}\label{SimpleEx}
    \cD\rho = \frac{1}{\tau}[\hat{\rho}_A\otimes\tr_A[\rho(t)]-\rho(t)]\ .
\end{equation}
(This dissipator can be written in the form \eqref{Lindblad2} and \eqref{dissform}, as we show in Appendix~\ref{LINFORM7}.)  
%This form \eqref{SimpleEx} of $\cD$ for a spin chain of $n$ sites with $A$ the spin at site $1$, or the spins at sites $1$ and $n$, is {\bf fix this.} considered in \cite{LPS}; an example is discussed below.  %There $\hat{\rho}_A$ is chosen to correspond to an ``equilibrium state'' of the density matrix of a single spin with a given reduced magnetic field $\tilde{h}=\beta h$ (or reduced chemical potential $\tilde\mu=\beta\mu$ if we consider the spin quanta as particles).  This can be thought of as being coupled to a reservoir at infinite temperature $\beta\rightarrow 0$, but with finite $\tilde{h}$ (finite $\tilde{\mu})$.  More generally, they study the case where $A$ consists of site $1$ and site $n$ of a chain of $n$ spins, with these end spins coupled to reservoirs with different $\tilde{h}$.

The question we now investigate is: Given a Lindbladian dissipator of the form \eqref{dissform} that acts only on part $A$ (the boundary) of our system, what can we say about a stationary density matrix of the full system, $\bar{\rho}$?  We first restrict our attention to cases where there is a steady state $\bar{\rho}$ that is a ``generalized Gibbs state'', meaning that it commutes with the Hamiltonian $H$: $[\bar{\rho},H]=0$.  This insures that $\bar{\rho}$ is also a stationary state of the system if we set $\cD=0$. %it is decoupled from the reservoir.  
It includes the standard Gibbs state $\rho_\beta=e^{-\beta H}/Z(\beta)$ as a special case.  We will show later that in some cases such a stationary state is unique.  Our main theorems and a corollary are:

%We prove below that if $\cD_A$ is ergodic on $A$ and $\bar{\rho}$ commutes with $H$, then $\bar{\rho}$ must be of the form 
%\begin{equation}\label{rhoVAVB}
%\bar{\rho}=\hat{\rho}_A\otimes\rho_B=\exp{[V_A+V_B]}~.
%\end{equation}

%It follows immediately that the Gibbs state $\rho_\beta$ with $\beta\neq 0$ can only be a stationary state if $H_{AB}=0$.  In that case $A$ and $B$ do not interact, so $\rho_B$ can be any density matrix that commutes with $H_B$: the dissipator relaxes $A$ to state $\hat{\rho}_A$, but $B$ remains isolated and autonomous, so has many possible steady states.

%\section{Formulation and proof of the main theorem}

\begin{thm}\label{main1}  
Let  $\cD$ %be a  quantum dynamical semigroup generator on $\cB(\cH_A\otimes\cH_B)$ where $\cD$ has 
have the form \eqref{dissform}.  We further assume that  $\cD_A$ is ergodic, %and non-transient, 
so on $\cH_A$ there is a unique density matrix $\hat\rho_A$ satisfying $\cD_A\hat{\rho}_A = 0$.  
Let $\overline{\rho}$ be a steady state solution of  \eqref{Lindblad1}.

Assume that $\bar{\rho}$ commutes with $H$.    %Assume that $\overline{\rho}$ commutes with $H$.  
Then there exists a density matrix $\rho_B$ on $\cH_B$ such that
\begin{equation}\label{rhoVAVB}
\overline{\rho} = \hat{\rho}_A\otimes \rho_B~.
\end{equation}
In particular, such a steady state $\overline{\rho}$ always satisfies $\tr_B\{\overline{\rho}\} = \hat{\rho}_A$. If we further assume that $\bar{\rho}$ is positive definite, then this implies that 
\begin{equation}\label{SEPZERO}
[\hat{\rho}_A,H_A]=0\ ,\quad [{\rho}_B,H_B]=0 \quad{\rm and}\quad [\bar{\rho},H_{AB}]=0\ .
\end{equation} 

 Conversely if $[\bar{\rho},H]\neq 0$ then $\bar{\rho}\neq\hat{\rho}_A\otimes\rho_B$ for any $\rho_B$.

\end{thm}

\begin{cor}\label{GibbsCor} 
Under the same assumptions made in Theorem~\ref{main1}, suppose that for some finite $\beta >0$, the Gibbs state
\begin{equation}\label{GIBBS}
\rho_\beta = \frac{1}{Z(\beta)}e^{-\beta H}\ 
\end{equation} 
is a steady state solution of \eqref{Lindblad1}. Then necessarily $H_{AB} =0$,
and for any density matrix $\rho_B$ on $\cH_B$ such that $[\rho_B,H_B]= 0$, $\hat{\rho}_A\otimes \rho_B$ is a steady state solution of \eqref{Lindblad1}. %In particular, the Gibbs state is not the unique steady state. %when $H_B$ is non-trivial.   If $H_{AB}=0$, then 
The dissipator relaxes $A$ to its unique steady state $\hat{\rho}_A$.  However, due to $H_{AB}=0$, subsystem $B$ remains isolated and autonomous so has many possible steady states.
\end{cor}

%We also now state a uniqueness condition that applies in our context of partial ergodicity. It is proved in the appendix.

We next turn to uniqueness and positivity. As already noted, the uniqueness of 
steady states is closely connected with the existence of 
positive definite steady states, as we recall below. In our setting, 
in which the Lindbladian acting at the boundary has a unique 
strictly positive steady state, a simple algebraic condition on 
the Hamiltonian $H$  is necessary and sufficient for any steady 
state $\bar\rho$ that commutes with $H$ to be the unique steady state, and in this case $\bar\rho$ is necessarily positive definite. 

%Next we state a general theorem due to Frigerio \cite{Frig78} about the uniqueness of the stationary state $\bar{\rho}$ of \eqref{Lindblad1}; see also \cite{yosh,zb} and references therein. %This is well known in the mathematical literature [ref?], although we do not know of a reference for an elementary proof in our finite-dimensional context. 
 %Therefore, we provide such a proof for the readers convenience. 

%\begin{thm}[Frigerio's Theorem]\label{Appthm} Suppose that \eqref{Lindblad1} has a $\cD\rho$ as in \eqref{Lindblad2},  and has at least one strictly positive steady state.
%Let
%\begin{equation}\label{Lindblad22}
%\cD\rho= %-i[K,\rho] + \sum_{\alpha=1}^n(L^{\dagger}_\alpha\rho L_\alpha-\frac{1}{2}\{L_\alpha L_\alpha^\dagger,\rho\})~
%\end{equation}
%be a Lindblad representation of $\cD$.
%Then there is a unique steady state density matrix $\overline{\rho}$ of \eqref{Lindblad1} if and only if
 %$H$ and $\{L_1,\dots,L_n\}$ are such that any operator $X$ that satisfies
%\begin{equation}\label{CommCond}
%[H,X]= 0 \quad {\rm and\ for \ all}\quad  \alpha, \quad [L_\alpha,X] = [L_\alpha^\dagger,X] = 0\ 
%\end{equation}
 %is a multiple of the identity.  This is equivalent to saying that $\%{H,L_\alpha,L_\alpha^\dagger\}$ generate all operators on $\cH$. 
%\end{thm}

%Considering our case where the dissipator acts only on the boundary, we obtain the following theorem: 

\begin{thm}\label{UNIQUECOR}
    Let $\cD$ be a Lindbladian dissipator of the form \eqref{Lindblad2} and \eqref{dissform}
%\begin{equation}\label{OURFORM}
%\cL\rho = -i[H,\rho] + (\cD_A\otimes\one_B)
%\end{equation}
    and suppose that $\cD_A$ is ergodic with a positive definite steady state $\hat \rho_A$ on $\cH_A$. %, and has a Lindblad form given by \eqref{Lindblad2}. 
    Suppose that $\bar\rho$ is a steady state for 
    \eqref{Lindblad1} that commutes with $H$.
    
    Then  $\bar\rho$ is the unique steady state of \eqref{Lindblad1}  if and only if the only traceless self-adjoint operator $X_B$ acting on $\cH_B$ such that $[H,\one_A\otimes X_B] = 0$ is $X_B = 0$, and moreover, in this case $\bar\rho$ is positive definite.

%If $\bar\rho$ is a steady state that commutes with $H$, but is not positive definite, then the steady state is not unique. Hence, whenever, there is a unique steady state $\bar\rho$ that commutes with $H$, it is necessarily positive definite. 
    
Furthermore, if there is a traceless %conserved 
operator $X_B$ on $\cH_B$ so that $[H,\one_A\otimes X_B]=0$ and thus the steady state is not unique, this also implies $[H_B, X_B]=0$ and $[H_{AB},\one_A\otimes X_B]=0$.
    \end{thm}

The proofs will follow below, after the following discussion:

%One might imagine that in the context of Theorem~\ref{main1}, the condition $[H,\overline{\rho}] =0$ means that the interaction described by $H_{AB}$ is not effective enough to ensue that \eqref{Lindblad1} has a unique steady state solution. However, this need not be the case: Another possibility is that $\hat{\rho}_A=e^{V_A}$ and $\rho_B=e^{V_B}$ so that
There are many examples, one of which we discuss below, for which both \eqref{rhoVAVB} and \eqref{SEPZERO} hold.  These include cases for which
\begin{equation}\label{rhov}
\bar{\rho}=e^V=e^{V_A}\otimes e^{V_B}~,
\end{equation}
%where $V=V_A\otimes\cI_B+\cI_A\otimes V_B$ (which commutes with $H$) is a finite operator that is not a sum of a multiple of $H$ and a multiple of the identity operator.  %The requirement that $\bar{\rho}$ commutes with $H$ Eqn. \eqref{SEPZERO} then implies that $[V_A,H_A]=0$, $[V_B,H_B]=0$, and $[V,H_{AB}]=0$.  include cases 
where $V=\lambda N + a\one$, where $N=N_A\otimes\one_B+\one_A\otimes N_B$ is the number of ``particles'' in the system, which is conserved by $H$ but is not conserved by $\cD$. $\one=\one_A\otimes\one_B$ is the identity operator on the full system, $V_A=\lambda N_A +a\one_A$ and $V_B=\lambda N_B +a\one_B$.  These ``particles'' can be the $z$ component of the magnetization for the spin chain examples \cite{LPS,TDKP}.  In \cite{AMTB} they are the excitations in a system of oscillators.  In these cases we have $\bar{\rho}\sim e^{\lambda N}$, which is the equilibrium Gibbs ensemble in the limit of infinite temperature, where the reduced chemical potential $\lambda=\beta\mu$ remains finite in the limit.  In order for $\bar{\rho}$ to be the unique steady state in these cases, we need $V_B$ to be unique, which requires that $H_{AB}\neq 0$ and that the only operator acting on $B$ that commutes with $H$ is the identity, as in Theorem \ref{UNIQUECOR}. Thus  $[H_{AB},\one_A\otimes V_B]\neq 0$: $H_{AB}$ moves particles between $A$ and $B$.

\noindent{Proof of Theorem~\ref{main1}}: Suppose that $\overline{\rho}$ is a steady state of \eqref{Lindblad1}.  Then since $ [H,\overline{\rho}]=0$, 
$\cD\overline{\rho} = 0$. %Let $\widehat{\cH}_A$ denote $\cB(\cH_A)$ equipped with the Hilbert-Schmidt inner product. 

Let $\{X_1,\dots,X_M\}$ be a complete orthonormal basis of the operators acting on ${\cH}_A$ (so that $M= d^2$ if $d$ is the dimension of $\cH_A$) consisting of
eigenvectors of $\cD_A^\dagger\cD_A$.  Since $\cD_A$ is ergodic, the nullspace of $\cD_A$ is spanned by $\hat{\rho}_A$.   Since $\langle \rho_A|\cD_A^{\dagger}\cD_A \rho_A\rangle=\langle\cD_A \rho_A|\cD_A \rho_A\rangle$ for any $\rho_A$, $\hat{\rho}_A$ also spans the nullspace of $\cD_A^\dagger\cD_A$.  We take its normalization in ${\cH}_A$ to be the first element, $X_1$, of our orthonormal basis.
Then for all $j > 1$, $\cD_A^\dagger\cD_A X_j = \sigma_j^2 X_j$ with $\sigma_j^2 > 0$. 

Then $\overline{\rho}$ has the expansion
\begin{equation}
\overline{\rho} = \sum_{j=1}^M X_j\otimes W_j\ ,
\end{equation}
where each $W_j$ acts on $\cH_B$.  Since $\cD\overline{\rho} = 0$ and $\cD=\cD_A\otimes\cI_B$ we have 
\begin{equation}
0 = \sum_{j=2}^M (\cD_A X_j)\otimes W_j\ .
\end{equation}
For $j>1$,  define $Y_j := \sigma_j^{-2}\cD_A X_j$. Note that 
\begin{eqnarray*}
\tr_{A}[Y_j^\dagger (\cD_A X_k)] &=& \langle  \sigma_j^{-2} \cD_A X_j,\cD_A X_k\rangle\\
&=& \sigma_j^{-2}\langle \cD_A^\dagger\cD_A X_j,X_k\rangle = \delta_{j,k}\ .
\end{eqnarray*}
Therefore, for each $k>1$,
\begin{equation}
0 = \sum_{j=2}^M \tr_{A}[(Y_k^\dagger\otimes\one_B)((\cD_A X_j)\otimes W_j)] =  W_k\ .
\end{equation}
The conclusion is that for some normalization constant $c$ 
\begin{equation}
\overline{\rho}=  c\hat{\rho}_A \otimes W_1
\end{equation}
where $\tr_B[cW_1] =1$ and $cW_1\geq 0$. Defining ${\rho}_B := cW_1$, we see that $\overline{\rho} = \hat{\rho}_A\otimes {\rho}_B$, which proves \eqref{rhoVAVB}.  Note that the assumption that $\bar{\rho}$ (and thus also $\hat{\rho}_A$) is positive definite was not used yet, so this part of the theorem (unlike \eqref{SEPZERO}) is also true even if $\bar{\rho}$ has null eigenvectors, as can occur when $\cD$ is approximating a zero-temperature bath. 

On the other hand, if $\bar{\rho}=\hat{\rho}_A\otimes\rho_B$ then using \eqref{darhoa} the r.h.s. of \eqref{Lindblad1} is just $-i[H,\bar{\rho}]$ which would have to be zero if $\bar{\rho}$ is stationary.

%{\bf Make a much more compact proof of (8) assuming that $\bar{\rho}$ is positive definite:}

%All of the spectral projections of $\overline{\rho}$ are polynomials in $\overline{\rho}$. Therefore, since $[\overline{\rho},H] =0$, $H$ commutes with all of the spectral projections of $\overline{\rho}$.  Note that $\overline{\rho} =\hat{\rho}_A\otimes {\rho}_B$ is not positive definite if and only if ${\rho}_B$ is not positive definite, and if $P_B$ denotes the orthogonal projection onto the range of ${\rho}_B$, then $P := \cI_A\otimes P_B$ is the orthogonal projection onto the range of $\overline{\rho}$, and hence $[H,P] =0$. Since clearly $[H_A,P] =0$, $[H_B,P] + [H_{AB},P] =0$. Taking $\tr_A$ of both sides,
%$$
%{\rm dim}(\cH_A)[H_B,P] = - \tr_A([H_{AB},\cI_A\otimes P_B]) = 0
%$$
%since $\tr_A[H_{AB}] =0$. This shows that
%\begin{equation}\label{SEPEQPRJ}
%[H_A,P] =0\ ,\quad [H_B,P] = 0\quad{\rm and}\quad [H_{AB},P] =0\ .    
%\end{equation}

%Define $V_A = \log \hat{\rho}_A$. Since $\cI_A\otimes {\rho}_B$ is strictly positive on the range of $P$, there is a self adjoint operator $V_B$ on $\cH_B$ (which we identify with $\cI_A\otimes V_B$ on $\cH$) such that $V_B =PV_BP$   and $\hat{\rho}_A \otimes {\rho}_B = e^{V_A+V_B} - P^\perp$.  Then since $H$ commutes with $P^\perp$, $H$ commutes with $e^{V_A+V_B}$, and hence with $\log(e^{V_A+V_B}) = V_A+V_B$. 

Now, to prove \eqref{SEPZERO}, we assume that $\bar{\rho}$ is positive definite, so we can define $V_A=\log{\hat{\rho}_A}$ and $V_B=\log{\rho_B}$. Since $H$ commutes with $\bar{\rho}=e^{V_A}\otimes e^{V_B}$, it commutes with $\log{\bar{\rho}} = V_A\otimes\one_B+\one_A\otimes V_B=V$.  Therefore
\begin{eqnarray}
0 &=& [V_A\otimes\one_B + \one_A\otimes V_B, H_A\otimes\one_B+ \one_A\otimes H_B + H_{AB}] \nonumber\\
&=&[V_A,H_A]\otimes\one_B + \one_A\otimes [V_B, H_B]\nonumber \\
&+& [V_A\otimes\one_B+\one_A\otimes V_B,H_{AB}]\ .\label{THREETERMS}
\end{eqnarray}
Apply the partial trace $\tr_B$ to each term on the right side.  First,
\begin{equation}
\tr_B\{[V_A,H_A]\otimes\one_B\} = [V_A,H_A]~{\rm dim}(\cH_B) \ .
\end{equation}
We claim that $\tr_B$ of all other terms in \eqref{THREETERMS} are zero so that $[V_A,H_A] =0$. To see this, 
$\tr_B\{\one_A\otimes[V_B,H_B]\}= 0$ since it is the trace of a commutator on $\cH_B$. Next, $\tr_B\{[V_A\otimes\one_B,H_{AB}]\} = [V_A,\tr_B\{H_{AB}\}] =0$ by our convention that $\tr_B\{H_{AB}\} = 0$.
Finally, $\tr_B\{[\one_A\otimes V_B,H_{AB}]\}= 0$ by the partial cyclicity of the partial trace; that is, $\tr_B\{V_B H_{AB}\} = \tr_B\{H_{AB}V_B\}$.
This proves that $[V_A,H_A] =0$, and the same reasoning using instead $\tr_A$ shows that 
$[V_B,H_B]= 0$. Then \eqref{THREETERMS} simplifies to 
$[V,H_{AB}]= 0$.  For each of these vanishing commutators, we then use the fact that $[C,D]=0$ implies $[e^C,D]=0$ for any two operators $C$, $D$, to prove \eqref{SEPZERO}.  %Since clearly $[H_A,$ Since $[V_A,H_A] =0$, $[\hat{\rho}_A,H_A]=0$. Since $\rho_B = e^V_B -P^\perp$, $[{\rho}_B,H_B]=[P^\perp,H_B]=0$ by \eqref{SEPEQPRJ}.
QED

\noindent{Proof of Corollary~\ref{GibbsCor}}:  By \eqref{GIBBS}, $\rho_\beta$ is
%\begin{equation}
%\log \rho_\beta = -\beta [H_A + H_B + H_{AB}] -\log Z\ ,
%\end{equation}
%which is not 
of the form \eqref{rhoVAVB} if and only if $H_{AB}=0$, given that $\tr_A\{H_{AB}\}=0$ and $\tr_B\{H_{AB}\}=0$.
QED

%By Theorem~\ref{main1}, applied with $\overline{\rho} = \rho_\beta$,
%\begin{equation}
%\log \rho_\beta = V_A + V_B\ .
%\end{equation}
%Therefore
%\begin{equation}
%H_{AB} = -(H_A + \frac1\beta V_A) -  (H_B + \frac1\beta V_B)   +\log Z\ .
%\end{equation}
%Then $\tr_A[H_{AB}]= 0$ becomes
%\begin{equation}
 %   (H_B + \frac1\beta V_B) =  \log Z -    \frac{1}{{\rm dim}(\cH_A)}\tr[(H_A + \frac1\beta V_A)]\ .
%\end{equation}
%Likewise,
%\begin{equation}
 %   (H_A + \frac1\beta V_A) =  \log Z -\frac{1}{{\rm dim}(\cH_B)}\tr[(H_B + \frac1\beta V_B)]\ .
%\end{equation}
%Therefore, $H_{AB}$ is a constant multiple of the identity, and is therefore zero by the convention that $\tr_A[H_{AB}]$ and $\tr_B[H_{AB}]$ both vanish.  The final statement is now clear. 
%\end{proof}

%The proof of Theorem~\ref{Appthm} is given below in Section~\ref{UNIQ}.

Before proving Theorem~\ref{UNIQUECOR} we recall a theorem of Frigerio \cite{Frig78} that we will use: %Recall the $\cD$, as defined in \eqref{Lindblad2} specifies the general Linbladian generator acting on operators on $\cH$. Frigerio's Theorem gives a criteion for the the uniqueness of steady states for the corresponding Lindblad equation.

\begin{thm}[Frigerio's Theorem]\label{Appthm} Suppose that the equation $\frac{d\rho}{dt}=\cD\rho$ for density %\eqref{GKSL} 
matrices on $\cH$ with $\cD$ given by \eqref{Lindblad2} %is true for at least one strictly positive density matrix $\rho$ on $\cH$.
%Suppose that \eqref{Lindblad1} has $H=0$, a $\cD\rho$ as in \eqref{Lindblad2}  and 
has at least one strictly positive steady state.
%Let
%\begin{equation}\label{Lindblad22}
%\cD\rho= %-i[K,\rho] + \sum_{\alpha=1}^n(L^{\dagger}_\alpha\rho L_\alpha-\frac{1}{2}\{L_\alpha L_\alpha^\dagger,\rho\})~
%\end{equation}
%be a Lindblad representation of $\cD$.
Then there is a unique steady state density matrix 
$\bar{\rho}$ with $\cD\bar{\rho}=0$ %is the unique density matrix with $\cD\rho=0$of \eqref{Lindblad1} 
if and only if $K$ and $\{L_1,\dots,L_n\}$ are such that any operator $X$ on $\cH$ that satisfies
\begin{equation}\label{CommCond}
[K,X]= 0 \quad {\rm and\ for \ all}\quad  \alpha, \quad [L_\alpha,X] = [L_\alpha^\dagger,X] = 0\ 
\end{equation}
 is a multiple of the identity.  This is equivalent to saying that $\{K,L_\alpha,L_\alpha^\dagger\}$ generate all operators on $\cH$. 
\end{thm}

Note that equation \eqref{Lindblad1} has the form considered in Frigerio’s Theorem if we simply replace $K$ by $(K+H)$, so that Frigerio’s Theorem may also be applied to equation \eqref{Lindblad1}.

%Note that equation \eqref{Lindblad1} has the form considered in Frigerio's Theorem if we simply replace $K$ by $K+H$, so that Frigerio's Theorem may also be applied to equation \eqref{Lindblad1}. 

\noindent{Proof of Theorem~\ref{UNIQUECOR}}:  Let $A $ be an operator on $\cH_A\otimes \cH_B$ such that $[\hat{K}\otimes\one_B,A]=0$ and for all $\alpha$: $[\hat{L}_\alpha\otimes \one_B,A] = [\hat{L}_\alpha^\dagger \otimes \one_B,A]= 0$, where $\hat{L_\alpha}$ and $\hat{K}$ are defined in \eqref{Lhat} and \eqref{Khat}.
Expand $A = \sum_\gamma W_\gamma\otimes E_\gamma$ where the $W_\gamma$ are operators on $\cH_A$ and the $E_\gamma$ are an orthonormal basis for operators on $\cH_B$.  Then $[\hat{L}_\alpha\otimes \one_B,A]=0$ becomes
\begin{equation}
\sum_{\gamma} [\hat{L}_\alpha,W_\gamma]\otimes E_\gamma = 0\ ,
\end{equation}
and since the $E_\gamma$ are orthonormal, $[\hat{L}_\alpha,W_\gamma] =0$ for each $\alpha,\gamma$.  A similar argument shows that $[\hat{L}_\alpha^\dagger,W_\gamma] =0$ for each $\alpha,\gamma$, and that $[\hat{K},W_\gamma] =0$ for each $\gamma$. 

Since $\cD_A$ is ergodic on $\cH_A$, by Theorem~\ref{Appthm},
the only operators on $\cH_A$ that commute with $\hat{K}$,   $\hat{L}_\alpha$,
and $\hat{L}_\alpha^\dagger$  for all $\alpha$ are multiples of the identity. Hence each $W_\gamma$ is of the form $W_\gamma = c_\gamma \one_A$ for some constant $c_\gamma$.  It follows that 
\begin{equation}
A = \one_A\otimes X \quad{\rm where}\quad X = \sum_{\gamma}c_\gamma E_\gamma ~.    
\end{equation}
Therefore, the only operators $A$ on $\cH_A\otimes \cH_B$ that satisfy \eqref{CommCond} of  Theorem~\ref{Appthm} are operators of the form
$\one_A\otimes X$ such that $[H,\one_A\otimes X]=0$. 
%Hence the steady state will be unique if and only if the only such operators $X$ are multiples of the identity.  

Now suppose that the only operators of the form $\one_A\otimes X$ such that $[H,\one_A\otimes X]=0$ are multiples of the identity.  
Let $\widetilde  \rho$ denote a steady state  that has maximal support, which exists by Theorem~\ref{ExMAX}.  If $\widetilde\rho$ is positive definite,  then by Frigerio's Theorem, $\widetilde\rho$ is the unique steady state, and so $\bar\rho = \widetilde\rho$ which is positive definite.

On the other hand, if  $\tilde\rho$ is not positive definite, then neither is any other steady state, including our steady state $\bar\rho$ that commutes with $H$. 
We claim that in this case, there would exist self-adjoint operators $X$ on $\cH_B$ other than multiples of the identity such that  $[H,\one_A\otimes X]=0$. Hence under our assumption  on operators satisfying  $[H,\one_A\otimes X]=0$, $\tilde\rho$ must be positive definite, and must be the unique steady state by Frigerio's Theorem, and hence equals $\bar\rho$.

To see this, note that by Theorem 2.1, $\bar\rho$ has the form $\bar\rho= \hat\rho_A\otimes\rho_B$. Since we assumed that $\hat\rho_A$ is positive definite, the projector $P$ onto the null space of $\bar\rho$ has the form $\one_A\otimes P_B$ where $P_B$ is the projector onto the null space of $\rho_B$. 
By hypothesis, $[H,\bar\rho]=0$. Then since all of the spectral projections of $\bar\rho$ are polynomials in $\bar\rho$,  $\one_A\otimes P_B$ is can be written as a polynomial, and hence it commutes with $H$. But then $X_B := P_B - c\one_B$, where $c$ is chosen to make $X_B$  traceless, is a non-zero traceless self-adjoint operator such that $\one_A\otimes X_B$ commutes with $H$, and therefore, if $\bar\rho$ is a degenerate (i.e. not positive definite) steady state commuting with $H$, it is not the unique steady state.

%Now suppose that $\bar\rho$ is a steady state that commutes with $H$, but is not positive definite. Since by Theorem~\ref{main1}, any steady state has the form $\bar\rho = \rho_B\otimes\hat\rho_A$. Since $\hat\rho_A$ is positive definite, the projector $P$ onto the kernel of $\bar\rho$ has the form $\one_A \otimes P_B$ where $P_B$ is the projector onto the kernel of $\rho_B$.  By hypothesis, $[H,\bar\rho]= 0$, and then since all of the spectral projections of $\bar\rho$ are polynomials in $\bar\rho$, $\one_A\otimes P_B$ is as well, and hence it commutes with $H$. But then $X_B := P_B- c\one$, where $c$ is chosen to make $X_B$ traceless, $X_B$ is a non-zero tracelss self-adjoint operator such that $\one_A\otimes X_B$ commutes with $H$, and therefore, if $\bar\rho$ is a degenerate steady state commuting with $H$, it is not the unique steady state.

To simplify the condition on solutions of $[H,\one_A\otimes X]=0$, observe, that since $H$ is self adjoint, $[H,\one_A\otimes X]=0$ if and only if $[H,\one_A\otimes X^\dagger]=0$,
and hence it suffices to consider self adjoint $X$. Finally since $\one_A\otimes X$ commutes with $H$ if and only if $\one_A\otimes(X - \tr[X]\one_B)$ commutes with $H$, we may freely assume $X$ to be traceless.  Thus, the steady state is unique if and only if the only traceless self-adjoint operator $X$ on $\cH_B$  such that 
$[H,\one_A\otimes X]=0$ is $X=0$.

Now suppose that the stationary state is not unique, so that there exists a non-trivial operator $X$ on $\cH_B$ such that $[H,\one_A\otimes X] =0$. 
Then $[H_B + H_{AB},\one_A\otimes X]=0$, and since $\tr_A[H_{AB}] =0$,
\begin{equation}
0 =\tr_A{[H_{AB},\one_A\otimes X]}= [H_B,X]\ , 
\end{equation}
 from which the rest follows.
QED

\section{Spin chain example}

%Cite more refs with similar or same examples.

%As a simple example, we now consider the oft studied 
The boundary driven XX (or XY) spin model on a chain of $\ell$ sites for which the dissipator is of the form \eqref{SimpleEx} is exactly solvable,  and the unique $\bar{\rho}$ is of the form \eqref{rhov}.  This model, and close relatives of it, are also presented in \cite{LPS,pros,zb} and references therein to illustrate various thoerems discussed in those papers. In this section we discuss this model as an illustration of Theorem~\ref{UNIQUECOR} for a $\cD_A$ of the form \eqref{SimpleEx}.  

After the Jordan-Wigner (JW) transformation its Hamiltonian has the form, c.f. eq. (15) in \cite{LPS},
\begin{equation}\label{HAMDEF}
    H=\sum_{j=1}^{\ell -1} (a^\dagger_ja_{j+1}+a^\dagger_{j+1}a_j)~,
\end{equation}
where $a_j$, $a^\dagger_j$ are the usual annihilation and creation operators of the JW fermions at site $j$.  As is well known, the particle number operator
\begin{equation}
    N=\sum_{j=1}^\ell a^\dagger_ja_j
\end{equation}
commutes with $H$.

Let $A$ be the first site of this chain, $j=1$, while $B$ is all the remaining sites.
Fix $\beta > 0$ and define
\begin{equation}\label{RHOADEF}
\hat{\rho}_A := \frac{1}{1+e^{-\beta}}(e^{-\beta}|1\rangle\langle 1|_A
+ |0\rangle \langle0|_A) = \frac{1}{1+e^{-\beta}} e^{-\beta a_1^\dagger a_1}\ .
\end{equation}
%Or in spin notation with any finite $\tilde{h}$ we can have:
%\begin{equation}\label{RHOADEF2}
%\hat{\rho}_A := \frac{1}{2}(\cI_A
%+\tanh{(\tilde{h}\sigma_A^{(z)}}) )\ ,
%\end{equation}
%where $\sigma^{(\alpha)}$ are the Pauli matrices, $\alpha=x,y,z$.

 Let $\cD=\cD_A\otimes\cI_B$ and let $\cD_A$ be the dissipator defined as in \eqref{SimpleEx} by
\begin{equation}\label{SimpleEx2}
    \cD_A\rho = \epsilon[\hat{\rho}_A\otimes\tr_A[\rho(t)]-\rho(t)]\ 
\end{equation}
 in terms of $\hat\rho_A$ as in \eqref{RHOADEF}. %, and let $\cD := \cD_A\otimes \one_B$.
Let $\ell\geq 2$, and let $H$ be the Hamiltonian defined in \eqref{HAMDEF}.
Define $\bar\rho$ to be the $\ell$-fold tensor product state
\begin{equation}
\bar\rho := \left(\frac{1}{1+e^{-\beta}}(e^{-\beta}|1\rangle\langle 1|
+ |0\rangle \langle0|)\right)^{\otimes \ell} 
\end{equation}
acting on $\cH_A\otimes\cH_B$ so that ${\displaystyle \bar\rho =  \left(\frac{1}{1+e^{-\beta}}\right)^\ell e^{-\beta N}}$.

Note  that $\bar\rho$ has the form $\hat\rho_A\otimes\rho_B$, so that $\cD\bar\rho =0$. Moreover, since $H$ commutes with $N$, and since $\bar\rho$ is a function of $N$,
 $[H,\bar\rho] =0$. Therefore $\bar\rho$ is a steady state of 
 \eqref{Lindblad1}.  
Since $\bar\rho$ is positive definite, one could apply Frigerio's Theorem to prove that $\bar\rho$ is the unique steady state -- there are many ways to treat this simple model. However, the work is especially simple using Theorem~\ref{UNIQUECOR} since we need only concern oursleves with $H$ and not the operators $L_\alpha$ and $L_\alpha^\dagger$ in the Lindblad description of $\cD_A$.

\noindent{Proof that $\bar\rho$ is the unique steady state via 
Theorem~\ref{UNIQUECOR}}: Since 
$\{n, a,a^\dagger,\one-n\}$
is an orthonormal basis for operators on 
$\cH_A$, we may expand  
\begin{eqnarray}
H = K_{1,1}\otimes n + K_{1,0}\otimes a + K_{0,1}\otimes a^\dagger + K_{0,0}\otimes (\one-n)\label{BLOCKFORM1}
\end{eqnarray}
and then write $H$ is the block matrix form
\begin{equation}\label{BLOCKFORM2}
H = \left[\begin{array}{cc} K_{1,1} & K_{1,0}\\ K_{0,1} & K_{0,0} \end{array}\right]\ 
\end{equation}
with  operators $K_{i,j}$ on $\cH_B$.

We will proceed by induction on $\ell$. For $\ell=2$, \eqref{BLOCKFORM2}
reduces to 
${\displaystyle H = \left[\begin{array}{cc} 0 & a\\ a^\dagger & 0 \end{array}\right]}$.
Likewise, the block form of $\one\otimes X$ is ${\displaystyle \one\otimes X = \left[\begin{array}{cc} X & 0\\ 0 & X \end{array}\right]}$.
Then $[H,\one\otimes X] =0$ becomes 
$\left[\begin{array}{cc} 0 & [X,a]\\ {[X,a^\dagger]} & 0 \end{array}\right]$
which reduces to $[a,X] =0$ and 
$[a^\dagger,X] =0$.  Any operator that commutes with both $a$ and $a^\dagger$ also commutes with $n$ and $\one-n$, and hence with everything. Therefore, any such operator $X$ is a multiple of the identity. Since $\tr[X] =0$, $X=0$.  %This proves the uniqueness for $N=2$.  
This proves uniqueness for $N=2$. 

 For $N>2$, let $X$ be self-adjoint on $\cH_B$ and such that $[\one\otimes X,H] =0$. We claim that then $X$ has the form $X= \one\otimes Y$
corresponding to the decomposition $\cH_B = \cH\otimes \cH^{\otimes N-2}$.

To see this, again write $H$ in the block form \eqref{BLOCKFORM2} with operators on $\cH_B$ as entries: 
\begin{equation}\label{BLOCKFORM3}
H = \left[\begin{array}{cc} K & a\otimes \one\\ a^\dagger\otimes \one & K\end{array}\right]\ ,
\end{equation}
where $a\otimes \one$,  $a^\dagger \otimes \one$ act on $\cH_B$ through its identification with $\cH\otimes \cH^{\otimes N-2}$, and where
${\displaystyle
K := \sum_{j=2}^{N-1} H_{j,j+1}}$.
Then $[H,\one\otimes X] =0$ is equivalent to
\begin{equation}
[K,X] = 0\ ,\quad [ a\otimes \one,X] =0 \quad{\rm and}\quad  [ a^\dagger \otimes \one,X] =0 \ .
\end{equation}
Now let $\{E_1,\dots, E_{2^{N-1}}\}$ be an orthonormal basis of operators on $\cH^{\otimes N-2}$. Then $X$ has a unique expansion
${\displaystyle
X = \sum_{j=1}^{2^{N-1}} W_j\otimes E_j
}$
where each $W_j$ is an operator on $\cH$. Then 
$0 = [ a\otimes \one,X] =\sum_{j=1}^{2^{N-1}} [a,W_j]\otimes E_j$ and 
$0 = [ a^\dagger\otimes \one,X] =\sum_{j=1}^{2^{N-1}} [a^\dagger,W_j]\otimes E_j$. It follows that for each $j$ $[a,W_j] = [a^\dagger,W_j]= 0$, and then  $W_j = c_j\one$ for some constant $c_j$.  Therefore
${\displaystyle X = \sum_{j=1}^{2^{N-1}} \one \otimes c_jE_j = \one\otimes Y}$
where ${\displaystyle Y = \sum_{j=1}^{2^{N-1}}  c_jE_j}$.

Now make the inductive assumption that this has been proved  for 
$N \leq M$; we shall show it is then true for $N=M+1$. 

Let $X$ be traceless and self adjoint  on $\cH_B = \cH^{\otimes M}$, and suppose that $\one\otimes X$ commutes with 
$H = \sum_{j=1}^{M} H_{j,j+1}$.  By what we proved just above, $X = \one\otimes Y$, where $Y$ is traceless and self adjoint on the last 
$M-1$ factors of $\cH$ in $\cH_B$.  Then $\one\otimes X = \one\otimes \one \otimes Y$, which evidently commutes with $H_{12}$. Therefore 
$[\one \otimes X,H] =0$ becomes 
\begin{equation}
[\one \otimes Y,H'] =0 \quad{\rm where}\quad H' = \sum_{j=2}^{M}H_{j,j+1}\ .  
\end{equation}
By the inductive hypothesis, $Y= 0$. 
QED

\begin{rem} Note that the form $\bar\rho = \hat{\rho}_A\otimes{\rho}_B$ of the unique steady state is independent of the parameter $\epsilon$, and this proves analytically that, as a function of $\epsilon$, the steady state does not converge to the Gibbs state as $\epsilon$ converges to zero, an issue discussed in \cite{TDKP}.     
\end{rem}

\if false
  
This is in stark contrast with classical dynamics, where one may take a chain of $N$ oscillators, coupled as usual through 
position-dependent interactions, and then drive the system with an Ornstein-Uhlenbeck process coupled only to the momentum variable at one 
end of the chain. Under appropriate conditions on the interactions, the result is that the corresponding Kolmogorov forward equation for the phase space density  is ergodic,  
and its unique steady state solution is the Gibbs distribution at a temperature determined by the Gaussian steady state of the Ornstein-Uhlenbeck process.
Theorem~\ref{summ} says that there is no analog of this in the quantum setting with finite-dimensional Hilbert spaces. 

By a theorem of Lindblad \cite{Lin76} and Gorini, Kossakowski and Sudarshan \cite{GKS76},  $\cL_A$ has the form
\begin{equation} 
\cL_A(X) = -i[K,X] + \sum_{i=1}^m L_i^{\phantom{\dagger}} XL_i^\dagger - \frac12(L_i^\dagger L_i^{\phantom{\dagger}} X + 
X L_i^\dagger L_i^{\phantom{\dagger}})
\end{equation}
where $K$ is self adjoint, and $\{L_1,\dots,L_m\} \subset \cB(\cH)$.  If $d_A$ denotes the dimension of $\cH_A$,  $m \leq d_A^2$.  
The self adjoint operator $K$ is sometimes called the {\em Lamb shift Hamiltonian}.  

The precise form of the Lindbladian generator depends on the model for the thermal bath and its coupling to the system. 
If these are specified so that they  satisfy certain reasonable conditions, a theorem of Davies \cite{Da74,SpLe77} gives the form of $\cL$.  
However, the equilibrium induced by the bath should not depend  much on the nature of the bath, and one might hope to be able to 
choose a suitable Lindbladian $\cL_A$, coupled only to $A$, % one site, 
so that for $\beta>0$ depending on the $\cL_A$, the Gibbs state
\begin{equation}
\rho_\beta = \frac{1}{Z}e^{-\beta H}\ 
\end{equation}
is the unique steady state of the full system ($A$ and $B$). However, in our quantum setting, this is impossible.

The analysis of this problem involves the ergodic properties of the driving Lindbladian equation.  We do not want to assume ergodicity of the quantum Markov semigroup generated by $\cL_A$.  Indeed, in the classical system described above, the driving Ornstien-Uhlenbeck process acts only on the momentum 
variables $p$ of the last oscillator, and not on the position variables $q$.  The next subsections recall the ergodic theory that is used in the more precise formulations of Theorem~\ref{summ}, beginning with the classical case to which we wish to compare. 

\subsection{Classical ergodic decompositions} In a classical setting that is analogous to the quantum setting that we investigate, consider a continuous time Markov process on a finite set $\cX$ with 
transition  matrix
$P_t(x_1,x_2)$, giving the probability of a transition from $x_1$ to $x_2$ in time $t$. There are two semigroups associated to these transition probabilities.
The first updates probability densities:
If $\rho(x)$ is a probability density on $x$ (with respect to counting measure) specifying the probabilities for the state at time $0$, the probability density $\rho_t(x)$  specifying the probabilities for the state at time $t$ is
\begin{equation}\rho(x,t) := P_t \rho(x) := \sum_{x'\in \cX}\rho(x')P_t(x',x)\ .
\end{equation}
Since $P_t$ is a semigroup,  it has the form $P_t = e^{tL}$ and $\rho(x,t) := e^{tL}\rho(x)$ solves the Kolmogorov forward equation:
\begin{equation}
\frac{\partial}{\partial t}\rho(x,t) = L \rho(x,t)\ .  
\end{equation}
The other is the dual semigroup given by
\begin{equation}
P^\dagger_t f(x) := f(x,t) = \sum_{x'\in \cX}P(x,x')f(x')\ ,
\end{equation}
which can be written as $P_t^\dagger = e^{tL^\dagger}$ in terms of its generator $L^\dagger$. If we introduce the Hilbert space $L^2(\cX)$ with counting measure on $\cX$, $P^\dagger_t$ is the adjoint of $P_t$, and then $L^\dagger$ is the adjoint of $L$.  Defining $f(x,t) := P^\dagger_tf(x)$, $f(x,t)$ solved the 
Kolmogorov backward equation
\begin{equation}
\frac{\partial}{\partial t}f(x,t) = L ^\dagger f(x,t)\ .  
\end{equation}
Then
$$
\sum_{x\in \cX}\rho(x,t)f(x) = \sum_{x\in \cX}\rho(x)f(x,t)\ .
$$

We briefly recall the ergodic decomposition of this stochastic process to give clear context to the description of its less familiar quantum analog.

By construction, $P_t^\dagger 1 = 1$ for all $t$, and hence $L^\dagger 1 = 0$.   Let $\cC$ denote the nullspace of $L^\dagger$, and suppose the dimension of $\cC$ is $m$.  Suppose there is at least one invariant density $\rho$ such that $\rho(x) > 0$ for all $x$. This ensures that there are no transient states. 
Then $\cX$ decomposes into a disjoint union 
\begin{equation}
\cX = \bigcup_{j=1}^m \cX_j \ ,
\end{equation}
such that  if we define the functions 
\begin{equation}
1_{\cX_j}(x) = \begin{cases} 1 & x\in \cX_j\\ 0 & x\notin \cX_j\end{cases}\ ,
\end{equation}
this orthogonal set of functions $\{1_{\cX_1},\dots,1_{\cX_m}\}$ spans $\cC$.  Then it is clear that $\cC$   is closed 
under multiplication as well as linear operations; that is, it is an algebra. The associated stochastic process 
leave each of the sets $\cX_j$ invariant, and there is exactly one invariant density supported on each $\cX_j$.   
These invariant densities $\{\rho_1,\dots,\rho_m\}$ are a basis for  the nullspace of $L$.  The decomposition
of the process into these components is called its ergodic decomposition.

The commutative case is the special case in which ${\rm dim}(\cK_1^{(j)}) =1$ for each $j$. Hence the following theorem subsumes Theorem~\ref{main1}, but is in fact a simple extension of it. 

Thus we have found that it is possible for $\rho_\beta$ to be a steady state of $\cL$, but only if $H$ and thus $\rho_\beta$ have some very special decoupled structure.  For such an $H$, what happens for a generic initial state $\rho(t=0)$?  Within each block, the Lindbladian is ergodic on $\cH_A^{(j)}$, so it will relax the state of $A$ in that block to $\omega_j$ in the long time limit.  In doing so, it will also fully dissipate any coherences between blocks, so the long time $\rho(t\rightarrow\infty)$ will be asymptotically block diagonal.  But within each block, the state on $B$ will evolve unitarily due to $H_B^{(j)}$, so in general will not go to a steady state; system $B$ is fully decoupled from the Lindbladian.  Thus although $\rho_\beta$ is a steady state, it is not unique, and generic initial states do not even go to a steady state at all.

We now extend this argument to the general case: 

\begin{thm}\label{main2} Let  $\cL$ be the generator of a quantum dynamical semigroup on $\cB(\cH)$ such that  $\rho_\beta$ is a steady state of  \eqref{lind1}.
Then the nullspace of $\cL^\dagger$ is a von Neumann subalgebra of $\cB(\cH)$. Let $\cZ$ denote the center of this algebra, and let 
$\{P_1,\dots,P_m\}$ be a set of mutually orthogonal projections spanning $\cZ$. For $j=1,\dots,m$, let $\cK_j$ denote the range of $P_j$, and let 
$\cK_j :=   \cK_1^{(j)}\otimes \cK_2^{(j)}$ and $\{\omega_1,\dots,\omega_m\}$ be such that every steady state solution $\rho_\infty$  of  $\cL\rho = 0$
has the form \eqref{STESTA1} for some density matrices $\{\rho_1,\dots,\rho_m\}$. 

 Then the Hamiltonian $H$ necessarily has the form 
\begin{equation}\label{DECOMP2}
H = \sum_{j=1}^m\left(  
H_j\otimes \one_{\cK_2^{(j)}} -\one_{\cK_1^{(j)}}\otimes \beta^{-1}\log \omega_j\right)\ 
\end{equation}
where $H_j$ is some self adjoint operator on $\cK_1^{(j)}$.  Again, while $\rho_\beta$ is a steady state solution of \eqref{lind1}, there are infinitely many others. 
\end{thm}

\begin{proof}  As before the partial trace of $\rho_\beta$ will be a faithful steady state for $\cL$, so that $\cC$ is an algebra.  Arguing exactly as above, we conclude that
\begin{equation}
\rho_\beta = \sum_{j=1}^m W_j\otimes \omega_j
\end{equation}
where each $W_j$ is strictly positive on $\cK_1^{(j)}$.  Then taking the logarithm as before yields \eqref{DECOMP2}. 
Conversely, it is clear that when $H$ has this structure, $\rho_\beta$ is a steady state of \eqref{lind1}.  Note that here the factor of $\cH$ that we had been separating out and labeling $\cH_B$ is now included in all of the $\{\cK_1^{(j)}\}$.
\end{proof}

In conclusion, whenever the Gibbs state $\rho_\beta$ is a steady state solution of \eqref{lind1}, the degrees of freedom on which the 
Lindbladian acts in a non-trivial way are completely uncoupled from the remaining degrees of freedom, and the Gibbs state is never 
the unique steady state solution of \eqref{lind1}.
This proves Theorem~\ref{main2}.

\fi

\section{Existence of steady states with maximal support}\label{S4}

We give here a simple proof of the existence of a stationary state of \eqref{Lindblad1} which yields some additional information that is used here.  Many proofs of existence of steady states invoke fixed point theorems; e.g., the Markov-Kakutani Fixed Point Theorem  in \cite{D70} in a general infinite dimensional setting, and the Brower fixed Point Theorem in \cite{ZB04} in a finite dimensional setting. The mean ergodic theorem provides a more constructive approach and additional information. 

\begin{thm}\label{ExMAX} For a $d$ dimensional Hilbert space $\cH$, the equation
\eqref{Lindblad1} has at least one steady state solution. Moreover, there exists a steady state solution $\bar\rho$  that has maximal support in the sense that if $\rho$ is any steady state solution, then 
\begin{equation}\label{MAxSupp}
\rho \leq d \bar\rho \ .   
\end{equation}  
\end{thm}

\noindent{proof:} Let 
$\cL\rho := -i[H,\rho] +\cD\rho$ as in \eqref{Lindblad1}. Then 
each $e^{t\cL}$, $t>0$ is 
completely positive and trace preserving. As a consequence, 
 by a mean ergodic theorem of Lance \cite{Lance76},
for any operator $A$ on $\cH$, the limit
\begin{equation}\label{ERGLIM}
   \lim_{T\to\infty}\frac1T \int_0^T e^{t\cL}(A){\rm d}t := \mathcal{P}_\cL(A) 
\end{equation}
exists.  (In our finite dimensional setting, all topologies are equivalent,so the sense of convergence is immaterial.)  It is clear from the definition that for all $t$, $e^{t\cL}\mathcal{P}_\cL(A) = \mathcal{P}_\cL(A)$. Furthermore, since $\mathcal{P}_\cL$ preserves positivity and traces, if $A$ is any density matrix, then $\mathcal{P}_\cL(A)$ is a density matrix. This proves existence. 

Next, define the density matrix $\rho_0$ by $\rho_0 := \frac1d \one_\cH$, and define
\begin{equation}\label{ERGLIM2}
 \overline{\rho} := \mathcal{P}_\cL(\rho_0)\ .   
\end{equation}
Then $\bar\rho$ is a steady state. Now let $\rho$ be any other steady state. Since $\rho \leq \one_\cH$, $\rho \leq d \rho_0$, and then for each $t$, $\rho = e^{t\cL}\rho \leq d e^{t\cL}\rho_0$ so that \eqref{MAxSupp} is satisfied.     
QED

We remark that in our finite dimensional setting, the theorem of Lance has an elementary proof using the Jordan canonical form and a well-known contractive property  of trace preserving completely positive operators.

\section{Uniqueness}\label{UNIQ}  

Frigerio's theorem (\ref{Appthm}) \cite{Frig78} gives a general if and only if result for uniqueness of the stationary solution $\bar{\rho}$ of \eqref{Lindblad1} once we know the existence of a positive definite $\bar{\rho}$.  The latter requirement is essential, as pointed out in \cite{zb}.  An avoidance of this requirement is given by Yoshida \cite{ZB04,yosh} who proved that a sufficient condition for uniqueness of $\bar{\rho}$ is that the Lindbladian $\cL$ is such that all operators in $\cH$ are linear combinations of products of the operators in the set $\{H-\frac{i}{2}\sum_\alpha L_\alpha^\dagger L_\alpha, L_\alpha\}$ (all $\alpha$).  This set generally contains fewer operators than the set used by Frigerio.

Theorem (\ref{UNIQUECOR}) gives necessary and sufficient conditions for uniqueness for the case when the dissipator $\cD$ has the form \eqref{dissform} and $\cD_A$ is ergodic.  We do not require the {\it a priori} existence of a positive definite $\bar{\rho}$ but find the conditions for uniqueness and strict positivity of a $\bar{\rho}$ of the form $\hat{\rho}_A\otimes\rho_B$ which commutes with $H$.  Our conditions also ensure that when there exists a unique $\bar{\rho}$ then it is strictly positive.

\section{Acknowledgements}

We thank Abhishek Dhar, Manas Kulkarni and Hironobu Yoshida for discussions.  E.A.C. was partially supported by U.S. National Science Foundation grant  DMS-2055282.  D.A.H. was partially supported by U.S. National Science Foundation QLCI grant OMA-212075.

\appendix

\section{Appendix A: Non-ergodic $\cD_A$}\label{NONERG}

In this appendix, we explain how quantum ergodic decompositions may be used to extend Theorem~\ref{main1} to the case in which $\cD_A$ is not assumed to be ergodic

Let $\cH$ be a finite dimensional Hilbert space and let $\cL$ be the generator of a quantum 
dynamical semigroup $e^{t\cL}$ on operators on $\cH$ so that each $e^{t\cL}$ is completely positive and trace preserving. Then $\cL^\dagger$ is the 
generator of a quantum Markov semigroup $(e^{t\cL})^\dagger := e^{t\cL^\dagger}$. That is, for each $t$,  $e^{t\cL^\dagger}$ is completely positive with the property that $e^{t\cL^\dagger}\cI_\cH = \cI_\cH$. Because of this last property, $\cL^\dagger\cI_\cH = 0$. 

Let $\cC$ denote 
the null space of $\cL^\dagger$.  
Suppose that there exists at least one strictly positive steady state; that is, at least one strictly positive density matrix such that 
$\cL\rho = 0$.  Then the theorem of  Frigerio \cite{Frig78} says that $\cC$ is not just a vector space of operators on $\cH$; it is also closed under multiplcation and taking Hermitian adjoints, and evidently it contains $\cI_\cH$. This makes it a a von Neumann algebra.  Let $\cZ$ denote $\cC\cap \cC'$ where $\cC'$ is the commutant of $\cC$. This is a commutative von Neumann algebra callled the {\em center} of $\cC$. 
 
Every  commutative von Neumann algebra on a finite dimensional Hilbert space $\cH$  has the following  simple structure (see, e.g., \cite{L99}):  There is a set
$\{P_1,\dots,P_m\}$ of mutually orthogonal projections summing to $\cI_\cH$
whose complex span is the algebra.

The projectors  $\{P_1,\dots,P_m\}$ provide the basis for an {\em ergodic decomposition} of $e^{t\cL^\dagger}$.
Let $\cH_j$ denote the range of $P_j$ so that 
\begin{equation}
\cH = \bigoplus_{j=1}^m \cH_j\ .
\end{equation}

The following theorem is proved in \cite{L99,BN}: Each of the Hilbert spaces
$H^{(j)}$ has a factorization $\cH^{(j)} = \cK^{(j)}_\ell \otimes \cK^{(j)}_r$, determined by the generator $\cL$, where either of these factors may be, but neither need be,  one dimensional.  There is a set of $m$ density matrices on the ``right'' factors $\cK^{(j)}_r$, $\{\omega_1,\dots,\omega_m\}$ such that a density matrix $\rho$ on $\cH$
satisfies $\cL\rho =0$ if and only if 
it has the form
\begin{equation}\label{STESTA1}
\rho = \sum_{j=1}^m p_j \rho_j\otimes \omega_j
\end{equation}
where each $\rho_j$ is {\em any} density matrix on $\cK_1^{(j)}$ and the $p_j's$ are probabilities.  

The ergodic case is that in which $m=1$ and $\cK^{(1)}_\ell$ is one dimensional so that $\cH = \cK^{(1)}_r$ and then $\omega_1$ is the unique steady state. 

If we relax the assumption that $\cD_A$ is ergodic with a 
strictly positive steady state to only the assumption that 
$\cD_A$ has at least one strictly positive steady state, so that every steady state for $\cD_A$ has the form \eqref{STESTA1}, then the method of proof of Theorem~\ref{main1} can be used to prove that every steady state $\bar\rho$ of \eqref{Lindblad1} that commutes with $H$ has an expansion of the form \eqref{STESTA1} where now $\rho_j$ is a density matrix on 
$\cK^{(j)}_r\otimes \cH_B$: In this non-ergodic case, the steady states that commute with $H$ are a direct sum of components that again factor as tensor products.  Finally, if $\cD_A$ does not have any positive definite steady state. let $\bar\rho_A$  be a steady state of maximal support, as in Theorem~\ref{ExMAX}, and let $\cK_A$ be the subspace of $\cH_A$ that supports $\bar\rho_A$. (That is, $\cK_A$ is the orthogonal complement of the null spaces of $\bar\rho_A$). Let $P_A$ be the orthogonal projection onto $\cK_A$. Then for any operator $X$ on $\cH_A$
\begin{equation}
e^{t\cL}(P_AXP_A) = P_Ae^{t\cL}(P_AXP_A)P_A~, 
\end{equation}
so that the Lindbladian evolution may be restricted to operators on $\cK_A$, and then it has a strictly positive steady state (but a different Lindbladian description in terms of operators $L_\alpha$ now acting on $\cK_A$ instead of $\cH_A$.).  The above consideration apply to this reduced system, in which a ``transient part'' has been discarded. The transient part is irrelevant as far the the structure of steady sates is concerned.

\section{Appendix B: The Lindblad form of \eqref{SimpleEx}}\label{LINFORM7}
 
Here we show that \eqref{SimpleEx} can be written in the form \eqref{Lindblad2} and \eqref{dissform}.  %We need to decide where it goes in the paper. This example will be investigated below, and hence we explain how $\cD$ may be written in the form $\cD = \cD_A\otimes \one_B$ where $\cD_A$ has the form specified in \eqref{Lindblad2}. Take $K=0$ in \eqref{Lindblad2}. 

Let 
$\{|1\rangle,\dots,|d\rangle\}$ be an orthonormal basis for $\cH_A$, and for each $1\leq i,j \leq d$, define
\begin{equation}
 L_{i,j} =  \hat{\rho}_A^{1/2}|i\rangle\langle j|~. 
\end{equation}
Then $\sum_{1 \leq i,j\leq d}L_{i,j}^\dagger L_{i,j} =  \sum_{1 \leq i,j\leq d}\langle i|\hat{\rho}_A|i\rangle  |j\rangle\langle j| = \one_A$, and for any 
$X\in \cB(\cH_A)$,
\begin{eqnarray*}
\sum_{1 \leq i,j\leq d}L_{i,j} X L_{i,j}^\dagger
&=& \frac1d \sum_{1 \leq i,j\leq d} \langle j|X |j\rangle \hat{\rho}_A^{1/2}|i\rangle\langle i|\hat{\rho}_A^{1/2} 
\\ &=& \tr[X] \hat{\rho}_A\ .
\end{eqnarray*}
Now define the Lindbladian superoperator $\cD_A$ on $\cH_A$ by
\begin{equation}
\cD_A(X) = \sum_{1 \leq i,j\leq d}\left(L_{i,j} X L_{i,j}^\dagger-\frac12\{ L_{i,j}^\dagger L_{i,j},X\}\right)\ .
\end{equation}
By the computations made just above, for any operator $X$ acting in $\cH_A$,  
\begin{equation}\label{REDUCED}
\cD_A(X) = \tr[X]\hat{\rho}_A - X~. 
\end{equation}
It is evident from \eqref{REDUCED} that the nullspace of $\cD_A$ is spanned by $\hat{\rho}_A$. Therefore, $\cD_A$ is ergodic.  It now follows that for all $Y\in \cB(\cH_A\otimes \cH_B)$,
\begin{equation}
\cD_A\otimes \cI_B(Y) = \hat{\rho}_A \otimes \tr_A[Y] -Y\ .
\end{equation}
Therefore, the generator $\cD$ in \eqref{SimpleEx} has the from $\cD = \cD_A\otimes \cI_B$ where $\cD_A$ has the canonical Lindbladian form.

\bibliographystyle{unsrtnat}
%\bibliography{main.bib}

\end{document}